\documentclass[a4paper,11pt,reqno]{amsart}
\usepackage[a4paper]{geometry}
\usepackage{amssymb,amsmath}
\usepackage{mathtools}
\usepackage{enumerate}
\usepackage{bookmark}
\usepackage{hyperref}
\usepackage[initials,backrefs]{amsrefs}

\numberwithin{equation}{section}
\theoremstyle{plain}
\newtheorem{theorem}{Theorem}

\newtheorem{lemma}{Lemma}

\theoremstyle{definition}
\newtheorem{definition}{Definition}
\newtheorem*{assi*}{(I) Short-range interaction}
\newtheorem*{assp*}{(P) Log-H\"older continuity condition}
\newtheorem*{dskn*}{$\dskn$}
\newtheorem*{dsknn*}{$\dsknn$}
\theoremstyle{remark}

\newcommand{\prob}[1]{\DP\left\{#1\right\}}
\newcommand{\esm}[1]{\mathbb{E}\left[\,#1\,\right]}
\newcommand{\Bone}{\mathbf{1}}

\newcommand{\BC}{\mathbf{C}}

\newcommand{\BG}{\mathbf{G}}
\newcommand{\BH}{\mathbf{H}}

\newcommand{\BK}{\mathbf{K}}

\newcommand{\BP}{\mathbf{P}}
\newcommand{\BU}{\mathbf{U}}
\newcommand{\BV}{\mathbf{V}}

\newcommand{\BX}{\mathbf{X}}

\newcommand{\CJ}{\mathcal{J}}

\newcommand{\CR}{\mathcal{R}}
\newcommand{\CB}{\mathcal{B}}

\newcommand{\DN}{\mathbb{N}}
\newcommand{\DP}{\mathbb{P}}
\newcommand{\DR}{\mathbb{R}}
\newcommand{\DZ}{\mathbb{Z}}
\newcommand{\BDelta}{\mathbf{\Delta}}
\newcommand{\BPsi}{\mathbf{\Psi}}
\newcommand{\Bx}{\mathbf{x}}
\newcommand{\By}{\mathbf{y}}

\newcommand{\Bu}{\mathbf{u}}
\newcommand{\Bv}{\mathbf{v}}
\newcommand{\FB}{\mathfrak{B}}

\newcommand{\rA}{\mathrm{A}}
\newcommand{\rB}{\mathrm{B}}

\newcommand{\rR}{\mathrm{R}}
\newcommand{\rS}{\mathrm{S}}

\DeclareMathOperator{\card}{card}
\DeclareMathOperator{\diam}{diam}
\DeclareMathOperator{\dist}{dist}
\DeclareMathOperator{\supp}{supp}

\newcommand{\ee}{\mathrm{e}}

\newcommand{\comp}{\mathrm{c}}
\newcommand{\fui}{\mathrm{FI}}

\newcommand{\pai}{\mathrm{PI}}
\newcommand{\sep}{\mathrm{sep}}

\newcommand{\condI}{\mathbf{(I)}}
\newcommand{\condP}{\mathbf{(P)}}
\newcommand{\dsk}[1]{\mathbf{(DS.}k,#1,N\mathbf{)}}
\newcommand{\dsn}[2]{\mathbf{(DS.}#1#2,$\,N\mathbf{)}$}

\newcommand{\dskn}{\mathbf{(DS.}k,N\mathbf{)}}
\newcommand{\dsknn}{\mathbf{(DS.}k,n,N\mathbf{)}}

\newcommand{\dskonn}{\mathbf{(DS.}k+1,n,N\mathbf{)}}

\begin{document}
\title[localization for weakly interacting Anderson models]{Multi-particle localization for weakly interacting Anderson tight-binding models}

\author[T.~Ekanga]{Tr\'esor EKANGA$^{\ast}$}

\address{$^{\ast}$%
Institut de Math\'ematiques de Jussieu,
Universit\'e Paris Diderot,
Batiment Sophie Germain,
13 rue Albert Einstein,
75013 Paris,
France}
\email{tresor.ekanga@imj-prg.fr}
\date{\today}
\begin{abstract}
We establish the complete spectral exponential,  and the strong Hilbert-Schmidt dynamical localization for the one-dimensional multi-particle Anderson tight-binding model and for weakly interacting particles system. In other words, we  show stability of the one-dimensional localization from the single-particle to multi-particle systems with an arbitrary large  but finite number of particles and for sufficient weakly interacting models. The proof uses the multi-scale analysis estimates for multi-particle systems. The common probability distribution function of the random external potential in the Anderson model is assumed to be log-H\"older continuous, so the results apply to a large class of  Anderson models.
\end{abstract}
\maketitle
\section{Introduction}\label{sec:intro}
Localization for discrete multi-particle random Schr\"odinger operators was initially proved by Aizenmann and Warzel \cites{AW09,AW10} using the fractional moment method, by Chulaevsky and Suhov \cites{CS09a,CS09b} and by Klein and Nguyen \cite{KN13} using the multiscale analysis in the strong disorder regime. More recently, in \cite{FW15} the authors extended the fractional moment method on the continuous space. Some other strategies using different forms of the multi-scale analysis were recently developed by Chulaevsky \cites{C11,C12}. Similar results were obtained by Anne Boutet de Monvel et al. \cites{BCSS10,BCS11} for the multi-particle model with alloy-type external random potential in the continuum.

In \cite{AW09}, the authors assumed that the distribution function of the i.i.d. random variables is absolutely continuous with a bounded density satisfying another technical condition and  proved stability of localization from single-particle to weakly interacting multi-particle systems. More important, in \cite{AW09}, the authors obtained localization in the so-called Hausdorff-distance and the dynamical localization in the operator-norm. The present work based on the multi-scale analysis proves the spectral exponential localization in the max-norm and the dynamical localization in the Hilbert-Schmidt norm. In addition our method covers a large class of random Hamiltonians with log-H\"older continuous distributions.

We decided to address here only the lattice case. The continuum version of the work is the object of a forthcoming paper. We use an adaptation of the multi-particle multi-scale analysis of \cites{CS09b,C12} in the high disorder case and in \cites{E11,E12} in the low energy case and show the complete spectral exponential and the strong Hilbert-Schmidt dynamical localization for discrete models with log-H\"older continuous distributions in one dimension. The general strategy which uses a perturbation argument based on the resolvent identities for operators in Hilbert spaces is valid in an arbitrary dimension. We restricted ourselves in one dimension in order to use the well-known one dimensional complete localization for the single-particle system and then extend it to weakly interacting multi-particle systems.

Note that, there is no proof of the complete localization in higher dimension instead of the strong disorder regime. In fact, localization is only proved  at extreme energies.  In one dimension Carmona et al. \cites{CL90,CKM87} in the discrete case and Damanik et al. \cite{DSS02} in the continuum  proved  strong forms of the  Wegner and the  multi-scale analysis estimates for the single-particle Anderson-Bernoulli models. 

We substantially modified in this text, the multi-particle multi-scale analysis of \cite{CS09b} on the one hand and \cite{E12} on the other hand. Indeed, in the first paper, i.e., \cite{CS09b}, the authors used the fact that the disorder amplitude goes to infinity in the case of strong disorder and in the second paper, i.e., \cite{E12}, the multi-particle multi-scale analysis is done for energies in an unbounded interval of the form $(-\infty;E^*]$ for some constant $E^*>0$. In the present work, only compact intervals are of interest because we use the complete localization results from single-particle models.

The derivation of the spectral localization from the multi-scale analysis results is then obtained using the scheme proposed in \cites{E11,E12} which idea goes back to Fr\"ohlich et al. \cite{FMSS85} and von Dreifus and Klein \cite{DK89}. While for the proof of the strong dynamical localization we refer again to the paper \cite{E12}.  We recall that dynamical localization with methods relying on the multi-scale analysis was initially obtained by Germinet and De Bi\`evre \cite{GB98}, Damanik and Stollmann \cites{DS01,St01} and Germinet and Klein \cite{GK01}.

In one dimension, our main results for an arbitrary amplitude of the disorder and for sufficient weakly interacting systems are  Theorem  \ref{thm:exp.loc} (Anderson localization) and  Theorem \ref{thm:dynamical.loc} (strong dynamical localization).

In Section~\ref{sec:model.assumptions} we describe our multi-particle model and the main assumptions.  Section~\ref{sec:main.results} is devoted to the statement of the results. We establish in Section \ref{sec:initial.MSA} the initial scale length estimate for the one-dimensional multi-particle system with weak interaction. In Section \ref{sec:MSA.induction}, we develop the one-dimensional variable energy multi-particle multi-scale induction step on the lattice. The proofs of  the main results based on the multi-particle multi-scale analysis bounds can be found in the paper \cite{E12} Section 5.

\section{The model, assumptions and the results}\label{sec:model.assumptions}

\subsection{The n-particle Hamiltonian on the lattice}
Define the two following norms on $\ell^2(\DZ^{D})$ for arbitrary $D\geq 1$: $|\Bx|=\max_{i=1,\ldots,D}|x_i|$ and $|\Bx|_1=|x_1|+\cdots+|x_D|$.  We consider a system of $N$-particles where $N\geq 2$ is finite and fixed. Let $d\geq 1$ and  $1\leq n\leq N$. We analyze random Hamiltonian $\BH^{(n)}_{h}(\omega)$  of the form
\begin{equation}\label{eq:hamiltonian}
\BH^{(n)}_{h}(\omega)=-\BDelta+\sum_{j=1}^nV(x_j,\omega)+h\BU=-\BDelta+\BV(\Bx,\omega)+h\BU,
\end{equation}
acting on $\ell^2((\DZ^{d})^n)\cong \ell^2(\DZ^{nd})$ with $h\in\DR$ and $\Bx\in(\DZ^d)^n$.  Above, $\BDelta$ is the $nd$-dimensional lattice nearest-neighbor Laplacian:
\begin{equation}           \label{eq:def.Delta}
(\BDelta\BPsi)(\Bx)
=\sum_{\substack{\By\in\DZ^{nd}\\|\By-\Bx|_1=1}}\left(\BPsi(\By)-\BPsi(\Bx)\right)
=\sum_{\substack{\By\in\DZ^{nd}\\|\By-\Bx|_1=1}}\BPsi(\By)-2dn\BPsi(\Bx),
\end{equation}
for $\BPsi\in\ell^2(\DZ^{nd})$ and $\Bx\in\DZ^{nd}$. $V\colon\DZ^d\times\Omega\to\DR$ is a random field relative to a probability space $(\Omega,\FB,\DP)$ and $\BU\colon(\DZ^d)^n\to\DR$ is the potential of inter-particle interaction. $\BV$ and $\BU$ act on $\ell^2(\DZ^{nd})$ as multiplication operators by functions $\BV(\Bx,\omega)$ and $\BU(\Bx)$ respectively.

Technically, we will prove in this paper stability of initial MSA bounds  from the one-dimensional single-particle estimates to the multi-particle estimates under sufficiently weak interaction and then  perform the  multi-particle multi-scale analysis for the weakly interacting multi-particle Anderson model and obtain the complete Anderson localization.

\subsection{Assumptions}

\begin{assi*}

Fix any $n=1,\ldots,N$. The potential of inter-particle interaction $\mathbf{U}$ is bounded and of the form
\[
\BU(\Bx)=\sum_{1\leq i<j\leq n}\Phi(|x_i-x_j|),\quad \Bx=(x_1,\ldots,x_n),
\]
where  $\Phi:\DN:\rightarrow\DR$ is a compactly supported function such that

\begin{equation}\label{eq:finite.range.k}
\exists r_0\in\DN: \supp \Phi\subset[0,r_0].
\end{equation}

\end{assi*}

Set $\Omega=\DR^{\DZ^d}$ and $\FB=\bigotimes_{\DZ^d}\mathcal{B}(\DR)$ where $\mathcal{B}(\DR)$ is the Borel sigma-algebra on $\DR$. Let $\mu$ be a probability measure on $\DR$ and define $\DP=\bigotimes_{\DZ^d}\mu$ on $\Omega$.

The external random potential $V\colon\DZ^d\times\Omega\to\DR$ is an i.i.d. random field
relative to  $(\Omega,\FB,\DP)$ and is defined by $V(x,\omega)=\omega_x$ for $\omega=(\omega_i)_{i\in\DZ^d}$.
The common probability distribution function, $F_V$,  of the i.i.d. random variables $V(x,\cdot)$, $x\in\DZ^d$ associated to the measure $\mu$ is defined by
\[
F_V: t \mapsto \prob{V(0,\omega)\leq t }.
\]
\begin{assp*}
The random potential $V:\DZ^d\times\Omega\rightarrow \DR$ is i.i.d. and almost surely bounded, i.e., there exists $M\in(0,\infty)$ such that $\supp \mu \subset[-M,M]$ where $\mu$ is the common probability distribution measure of the random field $\{V(x,\omega)\}_{x\in\DZ^d}$. Further, the corresponding probability distribution function $F_V$ is log-H\"older continuous:  More precisely,
\begin{align}\label{eq:assumption.3prime}
&s(F_V,\varepsilon) := \sup_{a\in\DR}(F_V(a+\varepsilon)-F_V(a))
\leq\frac{C}{|\ln\epsilon|^{2A}}\\
&\text{for some }C\in(0,\infty)\text{ and }A>\frac{3}{2}\times4^Np+9Nd.\notag
\end{align}
Note that this last condition depends on the parameter $p$ which will be introduced  in Section \ref{sec:Nparticle.scheme}.
\end{assp*}

Remark that, under the assumptions $\condI$ and $\condP$, we have that the spectrum of the one-dimensional multi-particle  Hamiltonian $\BH^{(N)}_{h}(\omega)$ satisfies:
\[
\sigma(\BH^{(N)}(\omega){h})\subset [-N(4d+M)-|h|\|\BU\|,N(4d+M)+|h|\|\BU\|] \quad \emph{a.s.}
\]
Therefore it suffices for our purposes to show the Anderson localization on the interval 
\[
 I=[-N(4d+M)-|h|\|\BU\|,N(4d+M)+|h|\|\BU\|].
\]

\subsection{The results}\label{sec:main.results}

\begin{theorem}\label{thm:exp.loc}
Let $d=1$. Under assumptions $\condI$ and $\condP$,
 there exists $h^*>0$ such that for any $h\in(-h^*, h^*)$ the Hamiltonian $\BH^{(N)}_h$, with interaction of amplitude $|h|$, exhibits complete Anderson localization, \emph{i.e.}, with $\DP$-probability one,
the spectrum of $\BH^{(N)}_h$  is pure point, and the eigenfunctions $\BPsi_i(\Bx,\omega)$ relative to eigenvalues $E_i(\omega)\in I$ are exponentially decaying at infinity:
\[
|\BPsi_i(\Bx,\omega)|\leq C_i\ee^{-c|\Bx|},
\]
for some positive constants $c$ and $C_i$.
\end{theorem}

Denote by  $\CB_1$ the set  of bounded measurable functions $f:\DR\rightarrow\DR$ such that $\|f\|_{\infty}\leq 1$. We now give our result on strong Hilbert-Schmidt dynamical localization.

\begin{theorem}\label{thm:dynamical.loc}
Let $d=1$. Under assumptions $\condI$ and $\condP$, there exists $h^*>0$ and $s^*>0$ such that for any $h\in(-h^*, h^*)$  any bounded Borel function $f:\DR\rightarrow\DR$, any bounded region $\BK\subset\DZ^{Nd}$ and any $s\in(-s^*,s^*)$ we have:
\begin{equation}\label{eq:weak.interaction.dynamical.loc}
\esm{\sup_{f\in\CB_1}\Bigl\| |\BX|^{\frac{s}{2}}f(\BH^{(N)}(\omega))\BP_{I}(\BH^{(N)}_h(\omega))\Bone_{\BK}\Bigr\|_{HS}^2}<\infty.
\end{equation}
\end{theorem}
where $(|\BX|\BPsi)(\Bx):=|\Bx|\BPsi(\Bx)$, $\BP_{I}(\BH^{(N)}(\omega))$ is the spectral projection of $\BH^{(N)}(\omega)$ onto the interval $I$. The rest of the paper has a significant overlap with the text of \cite{E12}. Hence, we will refer to \cite{E12} for  some details and proofs. 
 
\section{The multi-particle multi-scale analysis scheme and the Wegner bounds}\label{sec:Nparticle.scheme}
According to the general structure of the MSA, we work with lattice \emph{rectangles}. For $\Bu=(u_1,\ldots,u_n)\in\DZ^{nd}$, we denote by $\BC^{(n)}_L(\Bu)$ the $n$-particle cube, i.e,
\[
\BC^{(n)}_L(\Bu)=\left\{\Bx\in\DZ^{nd}:|\Bx-\Bu|\leq L\right\},
\]
and given $\{L_i: i=1,\ldots,n\}$, we define the rectangle
\begin{equation}          \label{eq:cube}
\BC^{(n)}(\Bu)=\prod_{i=1}^n C^{(1)}_{L_i}(u_i),
\end{equation}
where $C^{(1)}_{L_i}(u_i)$ are cubes of side length $L_i$ center at points $u_i$.
We  define the \emph{internal boundary} of the domain $\BC^{(n)}(\Bu)$ by
\begin{equation}\label{eq:int.boundary}
\partial^-\BC^{(n)}(\Bu)=\left\{\Bv\in\DZ^{nd}:\dist\left(\Bv,\DZ^{nd}\setminus \BC^{(n)}(\Bu)\right)=1\right\},       
\end{equation}
and its \emph{external boundary} by
\begin{equation}\label{eq:ext.boundary}
\partial^+\BC^{(n)}(\Bu)=\left\{\Bv\in\DZ^{nd}\setminus\BC^{(n)}(\Bu):\dist\left(\Bv, \BC^{(n)}(\Bu)\right)=1\right\}.
\end{equation}
The cardinality of the cube $\BC^{(n)}_L(\Bu)$ is $|\BC_L^{(n)}(\Bu)| := \card\BC_L^{(n)}(\Bu)=(2L+1)^{nd}$.
We define the restriction of the Hamiltonian $\BH_h^{(n)}$ to  $\BC^{(n)}(\Bu)$ by
\begin{align*}
&\BH_{\BC^{(n)}(\Bu),h}^{(n)}=\BH^{(n)}_h\big\vert_{\BC^{(n)}(\Bu)}\\
&\text{with simple boundary conditions on }\partial^+\BC^{(n)}(\Bu),
\end{align*}
i.e., $\BH^{(n)}_{\BC^{(n)}(\Bu),h}(\Bx,\By)=\BH^{(n)}_h(\Bx,\By)$ whenever $\Bx,\By\in\BC^{(n)}(\Bu)$ and $\BH^{(n)}_{\BC^{(n)}(\Bu),h}(\Bx,\By)=0$ otherwise.
We denote the spectrum of $\BH_{\BC^{(n)}(\Bu),h}^{(n)}$  by
$\sigma\bigl(\BH_{\BC^{(n)}(\Bu)}^{(n),h}\bigr)$ and its resolvent by
\begin{equation}\label{eq:def.resolvent}
\BG_{\BC^{(n)}(\Bu),h}(E):=\Bigl(\BH_{\BC^{(n)}(\Bu),h}^{(n)}-E\Bigr)^{-1},\quad E\in\DR\setminus\sigma\Bigl(\BH_{\BC^{(n)}(\Bu),h}^{(n)}\Bigr).
\end{equation}
The matrix elements $\BG_{\BC^{(n)}(\Bu),h}(\Bx,\By;E)$ are usually called the
\emph{Green functions} of the operator $\BH_{\BC^{(n)}(\Bu),h}^{(n)}$.

Let $m>0$ and $E\in\DR$ be given.  A cube $\BC_L^{(n)}(\Bu)\subset\DZ^{nd}$, $1\leq n\leq N$ will be  called $(E,m,h)$-\emph{nonsingular} ($(E,m,h)$-NS) if $E\notin\sigma(\BH^{(n)}_{\BC^{(n)}_{L}(\Bu),h})$ and
\begin{equation}\label{eq:singular}
\max_{\Bv\in\partial^-\BC_L^{(n)}(\Bu)}\left|\BG_{\BC_L^{(n)}(\Bu),h}(\Bu,\Bv;E)\right|\leq\ee^{-\gamma(m,L,n)L},
\end{equation}
where 
\begin{equation}\label{eq:gamma}
\gamma(m,L,n)=m(1+L^{-1/8})^{N-n+1}.
\end{equation}
Otherwise it will be called $(E,m,h)$-\emph{singular} ($(E,m,h)$-S).

Let us introduce the following.
\begin{definition}
Let $n\geq 1$, $E\in\DR$ and $\alpha=3/2$. 
\begin{enumerate}[\rm(A)]
\item
A cube $\BC_L^{(n)}(\Bv)\subset\DZ^{nd}$  is called $(E,h)$-resonant ($(E,h)$-R) if
\begin{equation} \label{eq:E-resonant}
\dist\Bigl[E,\sigma\bigl(\BH_{\BC_L^{(n)}(\Bv),h}^{(n)}\bigr)\Bigr]\leq\ee^{-L^{1/2}}.
\end{equation}
Otherwise it is called $(E,h)$-non-resonant ($(E,h)$-NR).
\item
A cube $\BC_L^{(n)}(\Bv)\subset\DZ^{nd}$ is called $(E,h)$-completely nonresonant ($(E,h)$-CNR), if it does not contain any $(E,h)$-R cube of size $\geq L^{1/\alpha}$. In particular $\BC^{(n)}_L(\Bv)$ is itself $(E,h)$-NR.
\end{enumerate}
\end{definition}

We will also make use of the following notion.

\begin{definition}\label{def:separability}
A cube $\BC^{(n)}_L(\Bx)$ is  $\CJ$-separable from $\BC^{(n)}_L(\By)$ if there exists a nonempty subset $\CJ\subset\{1,\cdots,n\}$ such that
\[
\left(\bigcup_{j\in \CJ}C^{(1)}_{L}(x_j)\right)\cap
\left(\bigcup_{j\notin \CJ}C_L^{(1)}(x_j)\cup \bigcup_{j=1}^n C_L^{(1)}(y_j)\right)=\emptyset.
\]
A pair $(\BC^{(n)}_L(\Bx),\BC^{(n)}_L(\By))$ is separable if $|\Bx-\By|>7NL$ and if one of the cube is $\CJ$-separable from the other.
\end{definition}
\begin{lemma}[\cite{E12}]\label{lem:separable.distant}
Let $L>1$.
\begin{enumerate}[\rm(A)]
\item
For any $\Bx\in\DZ^{nd}$, there exists a collection of $n$-particle cubes
$\BC^{(n)}_{2nL}(\Bx^{(\ell)})$ with $\ell=1,\ldots,\kappa(n)$, $\kappa(n)= n^n$ such that if
 $\By$ satisfies $|\By-\Bx|>7NL$ and
\[
\By \notin \bigcup_{\ell=1}^{\kappa(n)} \BC^{(n)}_{2nL}(\Bx^{(\ell)})
\]
then the cubes $\BC^{(n)}_L(\Bx)$ and $\BC^{(n)}_L(\By)$ are separable.
\item
Let $\BC^{(n)}_L(\By)\subset \DZ^{nd}$ be an $n$-particle cube. Any cube  $\BC^{(n)}_L(\Bx)$ with 
\[
|\By-\Bx|>\max_{1\leq i,j\leq n}|y_i-y_j| +5NL,
\]
 is $\CJ$-separable from
$\BC^{(n)}_L(\By)$ for some $\CJ\subset\{1,\ldots,n\}$.
\end{enumerate}
\end{lemma}

In our earlier work \cite{E12}  as well as in other previous papers in the multi-particle localization theory \cites{CS09b,BCSS10} the above notion was crucial in order to prove the Wegner estimates for pairs of multi-particle cubes via the Stollmann's Lemma. It is plain (cf. \cite{E12}, Section 4.1), that  sufficiently distant pairs of fully interactive cubes have disjoint projections and this fact combined with independence is used in that case to bound the probability of an intersection of events relative to those projections. There is a significant overlap of the induction step of the multi-scale analysis of the present text with its counterpart in our earlier paper  \cite{E12}.

\begin{theorem}[Wegner estimates \cite{E12}]\label{thm:Wegner}
 Assume that the random potential satisfies assumption $\condP$, then 
\begin{enumerate}

\item[\rm(A)]
for any $E\in\DR$
\begin{equation}\label{eq:cor.Wegner.2A}
\prob{\text{$ \BC^{(n)}_L(\Bx)$ is not $E$-CNR }}\leq L^{-p\,4^{N-n}},
\end{equation}

\item[\rm(B)]
\begin{equation}\label{eq:cor.Wegner.2B}
\prob{\text{$\exists E\in \DR:$ neither $\BC^{(n)}_{L}(\Bx)$ nor $\BC^{(n)}_{L}(\By)$ is $E$-CNR}} \leq L^{-p\,4^{N-n}},
\end{equation}
\end{enumerate}
where $p>6Nd$, depends only on the fixed number of particles $N$ and the configuration dimension $d$.
\end{theorem}

\section{The initial MSA bound for the weakly interacting multi-particle system}\label{sec:initial.MSA}

In  this Section we fix $d=1$, i.e., we consider here the one-dimensional multi-particle random Hamiltonian $\BH^{(N)}_{h}(\omega)=-\BDelta+\BV(\Bx,\omega)+h\BU(\Bx)$ acting in the Hilbert space $\ell^{2}(\DZ^N)$.
Now, we aim to prove stability of the MSA bounds from the single-particle lattice systems to multi-particle systems with sufficiently weak interaction in the  interval
\[
I=[-N(4d+M)|h|\|\BU\|,N(4d + M) + |h|\|\BU\|].
\]

Consider the one-dimesional single-particle Hamiltonian
\[
H^{(1)}(\omega)=-\Delta+V(\omega)
\]
with a non-constant i.i.d. random potential $V:\DZ\times \Omega\rightarrow \DR$. Recall that for the
one-dimensional Anderson models we have exponential localization for the eigenfunctions of random Hamiltonians with absolutely continuous probability distributions of bounded densities in \cite{KS80} for example and with H\"older or log-H\"older continuous distributions in \cites{CL90,DK89,FMSS85} and also with singular distributions such as the Bernoulli distributions in one dimension in \cites{CL90,CKM87}. We summarize this result in the following statement directly in  a form suitable for our forthcoming analysis (i.e., the MSA). Set

\[
\Upsilon_{C^{(1)}_{L_0}(u)}(x,y,I;\omega):=\sum_{E_j\in\sigma(H^{(1)}_{C^{(1)}_{L_0}(u)})\cap I}|\psi_j(x)\psi_j(y)|, \text{ for all $x,y\in C^{(1)}_{L_0}(u)$}.
\]
We have the following

\begin{theorem}[Single-particle localization]\label{thm:1p.loc}
There exists a constant $\widetilde{\mu}_1>0$ such that
\begin{equation}\label{eq:EF.corr.bound.1p}
\esm{ \Upsilon_{C^{(1)}_{L_0}(u)}(x,y,I;\omega)} \leq\ee^{-\widetilde{\mu}_1\, |x-y|},
\end{equation}
where $\{E_j,\psi_j\}_{j=1,\cdots,|C^{(1)}_{L_0}(u)|}$ are the eigenvalues and corresponding eigenfunctions of\\ $H^{(1)}_{C^{(1)}_{L_0}(u)}(\omega)$.
\end{theorem}

\subsection{The fixed energy  MSA bound for the n-particle system without interaction}\label{sec:non.int.initial.bound}
The main result of this subsection is Theorem \ref{thm:initial.bound.np} given below.
The proof of Theorem \ref{thm:initial.bound.np} relies on an auxiliary statement formulated below, Lemma \ref{lem:1p.NS.implies.np.NS}. We need to introduce first
\[
\{(\lambda^{(i)}_{j_i},\psi^{(i)}_{j_i}):j_i=1,\ldots,|C^{(1)}_{L_0}(u_i)|\},
\]
the eigenvalues and the corresponding eigenfunctions of $H^{(1)}_{C^{(1)}_{L_0}(u_i)}(\omega)$, $i=1,\ldots,n$. Then the  eigenvalues $E_{j_1\ldots j_n}$ of the non-interacting multiparticle random Hamiltonian\\ $\BH^{(n)}_{\BC^{(n)}_{L_0}(\Bu)}(\omega)$ are written as sums
\[
E_{j_1\ldots j_n}=\sum_{i=1}^n\lambda^{(i)}_{j_i}=\lambda^{(1)}_{j_1}+\cdots+\lambda^{(n)}_{j_n},
\]
while the corresponding eigenfunctions $\BPsi_{j_1\ldots j_n}$ can be chosen as tensor products
\[
\BPsi_{j_1\ldots j_n}=\psi^{(1)}_{j_1}\otimes\cdots\otimes\psi^ {(n)}_{j_n}.
\]
The eigenfunctions of finite volume Hamiltonians are assumed normalised.
\begin{theorem}\label{thm:initial.bound.np}
Let $1\leq n\leq N$ and $\widetilde{\mu}_1>0$ as in Theorem \ref{thm:1p.loc}. Consider $m^*=\min(\frac{1}{2^{N}12Nd},2^{-N-1}\widetilde{\mu}_1)$. Then for all
$E\in I$  and all $\Bu\in\DZ^{n}$:
\begin{equation}\label{eq:initial.bound.np}
\prob{\text{$\BC^{(n)}_{L_0}(\Bu)$ is $(E,m^*,0)$-S}}
\le \frac{1}{2} L_0^{-2p^*4^{N-n}},
\end{equation}
with $L_0$ large enough and $p^*>6Nd$.
\end{theorem}

The proof of Theorem \ref{thm:initial.bound.np} relies on the following auxiliary statement.

\begin{lemma}\label{lem:1p.NS.implies.np.NS}
Let be given $N\geq n\geq 2$, $m^*>0$, a cube
$\BC^{(n)}_{L_0}(\Bu)$ and $E\in \DR$.
Suppose that $\BC^{(n)}_{L_0}(\Bu)$ is $E$-NR, and for any operator
$H^{(1)}_{C_{L_0}(u_i)}$,
all its eigenfunctions $\psi_j$ satisfy
\begin{equation}\label{eq:psi.m.loc}
|\psi_j(u_i)\, \psi_j(u_i \pm L_0)| \le \ee^{-2\gamma(m^*,L_0,n)L_0}.
\end{equation}
Then $\BC^{(n)}_{L_0}(\Bu)$ is $(E,m^*,0)$-NS, provided that $L_0 \ge L_*(m^*,N,d)$.
\end{lemma}

\begin{proof}

Fix any $\By\in\partial^-\BC^{(n)}_{L_0}(\Bu)$. There exists $i\in[1,n]$ such that $|u_i-y_i|=L_0$.
Decompose the cube $\BC^{(n)}_{L_0}(\Bu)$ as follows:
$\BC^{(n)}_{L_0}(\Bu) = \BC^{(n-1)}_{L_0}(\Bu') \times \BC^{(1)}_{L_0}(u_i)$, $\Bu'\in\DZ^{(n-1)d}$.
In a similar way, we factorize every eigenfunction,
$\BPsi_k(\Bu) = \BPsi_{k'}(\Bu') \, \psi_i(u_i)$, and the respective eigenvalue,
$E_k = E_{\ne i} + \lambda^{(i)}_{j_i}$.
Now we have that
\[
\begin{aligned}
\BG_{\BC^{(n)}_{L_0}(\Bu)}(\Bu,\By;E)
= \sum_{k'} \BPsi_{k'}(\Bu') \BPsi_{k'}(\By')\; G^{(1)}_{C^{(1)}_{L_0}(u_i)}(u_i, y_i; E-E_{\neq i}).
\end{aligned}
\]
Here $\|\BPsi_{k'}\|_\infty \le 1$, since $\|\BPsi_k\|_2=1$, therefore,
\begin{align*}
\big| \BG_{\BC^{(n)}_{L_0}(\Bu)}(\Bu,\By;E) \big|
\le (2L_0+1)^{nd}\frac{\ee^{- 2\gamma(m^*,L_0,n)L_0 }}{ \ee^{- L_0^{1/2}}}
\le \ee^{- \gamma(m^*,L_0,n)L_0 },
\end{align*}
for $L$ large enough (depending on $ m^*, N,d$).
\end{proof}

\begin{proof}[Proof of Theorem \ref{thm:initial.bound.np}]
Introduce the events
\begin{align*}
&\mathcal{N} := \{\exists j=1,\ldots,n: \exists E_j\in\sigma(H^{(1)}_{C^{(1)}_{L_0}(u_j)}(\omega)): \text{ $|\phi_j(u_i)\phi_j(u_i\pm L_0)|>\ee^{-2\gamma(m^*,L_0,n)L_0}$} \},
\\
&\CR := \{ \BC^{(n)}_{L_0}(\Bu) \text{ is $E$-R }  \}.
\end{align*}
Then by Lemma \ref{lem:1p.NS.implies.np.NS}, Theorem \ref{thm:1p.loc} and theorem \ref{thm:Wegner} (A), we have:
\begin{align*}
\prob{\BC^{(n)}_{L_0}(\Bu) \text{ is $(E,m^*,0)$-S}}
&\le \prob{\mathcal{N}} + \prob{\CR},\\
&\leq  2\cdot n\cdot(2L_0+1)^{2nd}\frac{\esm{ \Upsilon_{C^{(1)}_{L_0}(u_i)}(u_i,u_i \pm L_0,I;\omega)}}{\ee^{-2\gamma(m^*,L_0,n)L_0}}+\prob{\CR}\\
&\leq \ee^{(-\widetilde{\mu}_1+2\gamma(m^*,L_0,n))L_0}+ L_0^{-4^N\,p}\leq \frac{1}{2}L_0^{-p^*\,4^{N-n}}
\end{align*}
for some $p^*>6Nd$, since  $2\gamma(m^*,L_0,n)< 2^{N+1}m^*\leq \widetilde{\mu}_1$.     
\end{proof}

In the next subsection, $m^*>0$ is the constant from Theorem \ref{thm:initial.bound.np}.
\subsection{The fixed energy  MSA bound for weakly interacting multi-particle systems}\label{sec:weak.int.initial.bound}
Now we derive the required initial estimate from its counterpart established for non-interacting  systems.

\begin{theorem}\label{thm:weak.interaction.fixed}
Let $1\leq n\leq N$. Suppose that the Hamiltonians $\BH_{0}^{(n)}(\omega)$ (without inter-particle interaction) fulfills the following condition: for all $E\in I$ and all $\Bu\in\DZ^{nd}$
\begin{equation}\label{eq:lem.weak.U.1}
\prob{\BC_{L_0}^{(n)}(\Bu)\text{ is $(E,m^*,0)$-S}}\leq\frac{1}{2} L_0^{-2p^*4^{N-n}},\qquad\text{with } 
p^*>6Nd.
\end{equation}
Then there exists $h^*>0$ such that for all $h\in(-h^*, h^*)$ the Hamiltonian $\BH_{h}^{(n)}(\omega)$, with interaction of amplitude $|h|$, satisfies a similar bound: there exist some $p>6Nd, m>0$ such that for all $E\in I$ and all $\Bu\in\DZ^{nd}$
\[
\prob{\BC_{L_0}^{(n)}(\Bu)\text{ is $(E,m,h)$-S}}\leq\frac{1}{2} L_0^{-2p\,4^{N-n}}.
\]
\end{theorem}

\begin{proof} First observe that the result of \eqref{eq:lem.weak.U.1} is proved in the statemeent of Theorem \ref{thm:initial.bound.np}.
Set
\[
\BG_{\BC_{L_0}^{(n)}(\Bu),h}(E)=(\BH^{(n)}_{\BC_{L_0}^{(n)}(\Bu),h}-E)^{-1}, \; h\in\DR.
\]
By definition, a cube $\BC_{L_0}^{(n)}(\Bu)$ is $(E,m^*,0)$-NS iff
\begin{equation}\label{eq:proof.lem.weak.U.0}
\max_{\By\in\partial^-\BC_{L_0}^{(n)}(\Bu)}\,\Bigl|\BG_{\BC_{L_0}^{(n)}(\Bu),0}(\Bu,\By;E)\Bigr|\leq\ee^{-m^*(1 + L_0^{-1/8})^{N-n+1}L_0}.
\end{equation}
Therefore, there exists sufficiently small $\epsilon>0$ such that
\begin{equation}\label{eq:proof.lem.weak.U.1}
\max_{\By\in\partial^-\BC_{L_0}^{(n)}(\Bu)}\,\Bigl|\BG_{\BC_{L_0}^{(n)}(\Bu),0}(\Bu,\By;E)\Bigr|\leq\ee^{-m(1+L_0^{-1/8})^{N-n}L_0}-\epsilon,
\end{equation}
where $m= m^*/2>0$. Since, by assumption, $p^*>6Nd$, there exists $6Nd<p<p^*$ and $\tau>0$ such that $L_0^{-2p4^{N-n}}-\tau>L_0^{-2p^*4^{N-n}}$. With such values $p$ and $\tau$, inequality \eqref{eq:lem.weak.U.1} with $p^*>6Nd$ implies
\begin{equation}\label{eq:lem.weak.U.p.prime}
\DP\{\BC_{L_0}^{(n)}(\Bu)\text{ is $(E,m^*,0)$-S}\}<\frac{1}{2}L_0^{-2p4^{N-n}}-\frac{1}{2}\tau.
\end{equation}
Next, it follows from the second resolvent identity that
\begin{equation}\label{eq:proof.lem.weak.U.2}
\|\BG_{\BC_{L_0}^{(n)}(\Bu),0}(E) - \BG_{\BC_{L_0}^{(n)}(\Bu),h}(E) \|
\leq |h|\, \|\BU\| \cdot \|\BG_{\BC_{L_0}^{(n)}(\Bu),0}(E)\| \cdot \|\BG_{\BC_{L_0}^{(n)}(\Bu),h}(E)\|.
\end{equation}
By Theorem \ref{thm:Wegner}, applied to Hamiltonians
$\BH^{(n)}_{\BC_{L_0}^{(n)}(\Bu),0}$ and $\BH_{\BC_{L_0}^{(n)}(\Bu),h}^{(n)}$, for any $\tau>0$ there is $B(\tau)\in(0,+\infty)$ such that
\begin{align*}
&\DP\bigl\{\|\BG_{\BC_{L_0}^{(n)}(\Bu),0}(E)\|\geq B(\tau)\bigr\}\leq\frac{\tau}{4}\,,\\
&\DP\bigl\{\|\BG_{\BC_{L_0}^{(n)}(\Bu),h}(E)\|\geq B(\tau)\bigr\}\leq\frac{\tau}{4}\,.
\end{align*}
Therefore,
\begin{align*}
&\DP\bigl\{\|\BG_{\BC_{L_0}^{(n)}(\Bu),0}(E)-\BG_{\BC_{L_0}^{(n)}(\Bu),h}(E)\|
\geq |h|\,\|\BU\|B^{2}(\tau)\bigr\}\\
&\qquad\leq\DP\bigl\{\|\BG_{\BC_{L_0}^{(n)}(\Bu),0}(E)\|\geq B(\tau)\bigr\}+\DP\bigl\{\|\BG_{\BC_{L_0}^{(n)}(\Bu),h}(E)\|\geq B(\tau)\bigr\}\\
&\qquad\leq 2\,\frac{\tau}{4}=\frac{\tau}{2}\,.
\end{align*}
Set $h^*:=\frac{\epsilon}{2\|\BU\|(B(\tau))^2}>0$. We see that if $|h|\leq h^*$, then $|h|\times\|\BU\|\times(B(\tau))^2\leq\frac{\epsilon}{2}\,$.
Hence,
\begin{equation}\label{eq.proof.lem.weak.U.3}
\DP\bigl\{\|\BG_{\BC_{L_0}^{(n)}(\Bu),0}-\BG_{\BC_{L_0}^{(n)}(\Bu),h}\|\geq\frac{\epsilon}{2}\bigr\}\leq 2\,\frac{\tau}{4}\,.
\end{equation}
Combining \eqref{eq:proof.lem.weak.U.1}, \eqref{eq:lem.weak.U.p.prime}, and
\eqref{eq.proof.lem.weak.U.3}, we obtain that for all $E\in I$
\begin{align*}
&\DP\bigl\{ \BC_{L_0}^{(n)}(\Bu) \text{ is $(E,m,h)$-S}\bigr\}\\
&\quad\leq\DP\bigl\{\BC_{L_0}^{(n)}(\Bu)\text{ is $(E,m^*,0)$-S}\bigr\}\\
&+\DP\bigl\{\|\BG_{\BC_{L_0}^{(n)}(\Bu),0}(E)-\BG_{\BC_{L_0}^{(n)}(\Bu),h}(E)\|\geq\frac{\epsilon}{2}\bigr\}\\
&\quad\leq\bigl(\frac{1}{2}L_0^{-2p4^{N-n}}-\frac{1}{2}\tau\bigr)+\frac{\tau}{2}=\frac{1}{2}L_0^{-2p'4^{N-n}}.\qedhere
\end{align*}
\end{proof}

\subsection{The variable energy MSA bound for weakly interacting multi-particle systems}
Here, we deduce from the fixed energy bound, the variable energy initial multi-scale analysis bound for the weakly interacting multi-particle system. 

Recall that in order to prove the complete Anderson localization for the perturbed Hamiltonian $\BH^{(n)}_h(\omega)$, we are concern with the interval
\[
I:=[-N(4d+M)|h|\|\BU\|;N(4d+M)|h|\|\BU\|],
\]
which contains its entire spectrum.
To do so, in fact we will prove localization in each compact interval $I_0$ of the following form: let $E_0\in I$ and $\delta:=\frac{1}{2}\ee^{-2L_0^{1/2}}(\ee^{-m_1L_0}-\ee^{-mL_0})$ where $0<m_1<m$ by definition. Set
\[
I_0:=[E_0-\delta;E_0+\delta].
\]
The result on the variable energy MSA is given below in
\begin{theorem}\label{thm:initial.var.energy}
Let $1\leq n\leq N$. For any $\Bu\in\DZ^{nd}$ we have
\begin{equation}
\prob{\exists E\in I_0: \text{$\BC^{(n)}_{L_0}(\Bu)$ is $(E,m_1)$-S}}\leq L_0^{-2p\,4^{N-n}},
\end{equation}
for some $m_1>0$.
\end{theorem}

\begin{proof}
Let $E_0\in I$. By the resolvent equation
\[
\BG_{\BC^{(n)}_{L_0}(\Bu),h}(E)=\BG_{\BC^{(n)}_{L_0}(\Bu),h}(E_0)+ (E-E_0)\BG_{\BC^{(n)}_{L_0}(\Bu),h}(E)\BG_{\BC^{(n)}_{L_0}(\Bu),h}(E_0).
\]
If $\dist(E_0,\sigma(\BH^{(n)}_{\BC^{(n)}_{L_0}(\Bu),h}))\geq \ee^{-L_0^{1/2}}$ and $|E-E_0|\leq \frac{1}{2} \ee^{-L_0^{1/2}}$, then $\dist(E,\sigma(\BH^{(n)}_{\BC^{(n)}_{L_0}(\Bu),h}))\geq \frac{1}{2} \ee^{-L_0^{1/2}}$.

If in addition, $\BC^{(n)}_{L_0}(\Bu)$ is $(E_0,m,h)$-NS and $\By\in\partial^-\BC^{(n)}_{L_0}(\Bu)$, then 
\[
|\BG_{\BC^{(n)}_{L_0}(\Bu),h}(\Bu,\By;E)|\leq \ee^{-m(1+L_0^{-1/8})^{N-n+1}L_0} + 2 |E-E_0| \ee^{2L_0^{1/2}}.
\]
Therefore, for $m_1=\frac{m}{2}$, if we put 
\[
\delta=\frac{1}{2}\ee^{-2L_0^{1/2}}(\ee^{-m_1(1+L_0^{-1/8})^{N-n+1}L_0}- \ee^{-m(1+L_0^{-1/8})^{N-n+1}L_0}),\quad I_{\delta}=[E_0-\delta,E_0+\delta],
\]
we have that 
\begin{align*}
&\prob{\text{$\exists E\in I_0$, $\BC^{(n)}_{L_0}(\Bu)$ is $(E,m_1,h)$-S}}\leq \prob{\text{$\BC^{(n)}_{L_0}(\Bu)$ is $(E_0,m,h)$-S}}\\
&\qquad + \prob{\dist(E_0,\sigma(\BH^{(n)}_{\BC^{(n)}_{L_0}(\Bu)}))\leq \ee^{-L_0^{1/2}}}\\
&\leq \frac{1}{2} L_0^{-2p4^{N-n}} + L_0^{-p4^{N}}<L_0^{-2p4^{N-n}}.
\end{align*}
We used Theorem \ref{thm:weak.interaction.fixed} to bound the first term and the Wegner estimate Theorem \ref{thm:Wegner} (A) to bound the other term.
\end{proof}

\section{Multiscale induction}\label{sec:MSA.induction}

In the rest of the paper, we assume that  $n\geq 2$ and  $I_0$ is the interval from the previous section. Some parts of this section overlap with the paper \cite{E12}.

Recall the following facts from \cite{E12}: Consider a cube $\BC^{(n)}_{L}(\Bu)$, with $\Bu=(u_1,\ldots,u_n)\in(\DZ^d)^n$.  We define
\[
\varPi\Bu=\{u_1,\ldots,u_n\},
\]
and 
\[
\varPi\BC^{(n)}_L(\Bu)=C^{(1)}_L(u_1)\cup\cdots\cup C^{(1)}_L(u_n).
\]
\begin{definition}
Let $L_0>3$ be a constant and $\alpha=3/2$. We define the sequence $\{L_k: k\geq 1\}$ recursively  as follows:
\[
L_k:=\lfloor L_{k-1}^{\alpha}\rfloor +1, \qquad \text{for all $k\geq 1$}.
\]
\end{definition}

Let $m>0$ a positive constant, we also introduce the following property, namely the multi-scale analysis bounds at any scale length $L_k$, and for any pair of separable cubes $\BC^{(n)}_{L_k}(\Bu)$ and $\BC^{(n)}_{L_k}(\Bv)$,
\begin{dsknn*}
\[
\prob{\exists E\in I_0: \text{$\BC^{(n)}_{L_k}(\Bu)$ and $\BC^{(n)}_{L_k}(\Bv)$ are $(E,m)$-S}}\leq L_k^{-2p4^{N-n}},
\]
where $p>6Nd$.
\end{dsknn*}

In both the single-particle and the multi-particle system, given the results on the multi-scale analysis property $\dsknn$ above one can deduce the localization results see for example the papers \cites{DK89,DS01} for those concerning the single-particle case and \cites{E12,CS09a} for multi-particle systems. We summarize below the exponential  decay bound of the eigenfunction correlators  of finite volumes Hamiltonians which is naturally deduced from the exponential decay bound of the eigenfunctions of the random Hamiltonian in the entire space $\DZ^{nd}$. Before, introduce

\begin{definition}\label{def:upsilon}
 For any cube $\BC^{(n)}_L(\Bu)$ and any $\Bx,\By\in\DZ^{nd}$, put
\[
\mathbf{\Upsilon}_{\BC^{(n)}_{L}(\Bu)}(\Bx,\By,I_0):=\sum_{E_j\in\sigma(\BH^{(n)}_{\BC^{(n)}_{L}(\Bu)})\cap I_0}|\mathbf{\psi}_j(\Bx)\mathbf{\psi}_j(\By)|,
\]
where $\{E_j,\psi_j\}_{j=1,\cdots,|\BC^{(n)}_{L}(u)|}$ are the eigenvalues and corresponding eigenfunctions of $\BH^{(n)}_{\BC^{(n)}_{L}(\Bu)}(\omega)$.
\end{definition}

\begin{theorem}\label{thm:EFC.bound}
For any $1\leq n'<n$, assume that property $\dsk{n',N}$ holds true forall $k\geq 0$, then there exists a constant $\widetilde{\mu}_{n'}>0$ such that for any cube $\BC^{(n')}_{L}(\Bu')$
\begin{equation}\label{eq:EFC.bound}
\esm{ \mathbf{\Upsilon}_{\BC^{(n')}_{L}(\Bu')}(\Bx,\By,I_0;\omega)} \leq\ee^{-\widetilde{\mu}_{n'}\, |\Bx-\By|},\qquad\text{ for all $\Bx,\By\in\BC^{(n')}_{L}(\Bu')$}.
\end{equation}
\end{theorem}

Recall  the constant $m_1>0$ from the previous section,  Theorem \ref{thm:initial.var.energy}. For the rest of the analysis, we will need the sequence $\{m_n: n\geq 1\}$ defined in
\begin{definition}\label{def:m_n}
Given $m_1>0$ and $n\geq 2$, define $m_n$ as follows:
\[
m_n=\min_{1\leq n'\leq n-1} \{m_1,\, 2^{-N-1}\widetilde{\mu}_{n'}\},
\]
where each $\widetilde{\mu}_{n'}$ is given in the statement of Theorem \ref{thm:EFC.bound}.
\end{definition}

\begin{definition}[fully/partially interactive]\label{def:diagonal.cubes}
An $n$-particle cube $\BC_L^{(n)}(\Bu)\subset\DZ^{nd}$ is
called fully interactive (FI) if
\begin{equation}\label{eq:def.FI}
\diam \varPi \Bu := \max_{i\ne j} |u_i - u_j| \le n(2L+r_0),
\end{equation}
and partially interactive (PI) otherwise.
\end{definition}

The following simple statement clarifies the notion of PI cube.

\begin{lemma}\label{lem:PI.cubes}[\cite{E12}]
If a cube $\BC_L^{(n)}(\Bu)$ is PI, then there exists a subset $\CJ\subset\left\{1,\dots,n\right\}$ with $1\leq\card\CJ\leq n-1$ such that
\[
\dist\left(\varPi_{\CJ}\BC_L^{(n)}(\Bu),\varPi_{\CJ^{\comp}}\BC_L^{(n)}(\Bu)\right)>r_0,
\]
\end{lemma}

If $\BC^{(n)}_L(\Bu)$ is a PI cube by the above Lemma, we can write it as 
\begin{equation}\label{eq:cartesian.cubes}
\BC_L^{(n)}(\Bu)=\BC_L^{(n')}(\Bu')\times\BC_L^{(n'')}(\Bu''),
\end{equation}
with 
\begin{equation}\label{eq:cartesian.cubes.2}
\dist\left(\varPi\BC_L^{(n')}(\Bu'),\varPi\BC_L^{(n'')}(\Bu'')\right)>r_0,
\end{equation}
where $\Bu'=\Bu_{\CJ}=(u_j:j\in\CJ)$, $\Bu''=\Bu_{\CJ^{\comp}}=(u_j:j\in\CJ^{\comp})$, $n'=\card \CJ$ and $n''=\card \CJ^{\comp}$. 

Throughout, when we write a PI cube $\BC^{(n)}_L(\Bu)$ in the form \eqref{eq:cartesian.cubes}, we implicitly assume that the projections satisfy \eqref{eq:cartesian.cubes.2}.
Let $\BC^{(n')}_{L_k}(\Bu')\times\BC^{(n'')}_{L_k}(\Bu'')$ be the decomposition of the PI cube $\BC^{(n)}_{L_k}(\Bu)$ and $\{\lambda_i,\varphi_i\}$ and $\{\mu_j,\phi_j\}$ be the eigenvalues and corresponding eigenfunctions of $\BH_{\BC_{L_k}^{(n')}(\Bu')}^{(n')}$ and $\BH_{\BC_{L_k}^{(n'')}(\Bu'')}^{(n'')}$ respectively. The eigenfunctions $\{\varphi_i\}$ and $\{\phi_j\}$ are  assumed to form a basis of the Hilbert spaces $\ell^{2}(\BC^{(n')}_{L}(\Bu'))$ and $\ell^{2}(\BC^{(n'')}_L(\Bu''))$ respectively.  Next, we  can choose the eigenfunctions $\BPsi_{ij}$ of $\BH_{\BC^{(n)}_{L_k}(\Bu)}(\omega)$ as tensor products:
\[
\BPsi_{ij}=\varphi_i\otimes\phi_j
\]
therefore, the set $\{\BPsi_{ij}\}$  also forms a basis of the tensor product of Hilbert spaces 
\[
\ell^{2}(\BC^{(n')}_L(\Bu'))\otimes\ell^{2}(\BC^{(n'')}_L(\Bu''))\cong\ell^2(\BC^{(n)}_L(\Bu)).
\]
The eigenfunctions appearing in subsequent arguments and calculation will be assumed normalized.

Now we turn to geometrical properties of FI cubes.
\begin{lemma}\label{lem:disjointness}[\cite{E12}]
Let $n\geq 1$, $L>2r_0$ and consider two FI cubes $\BC_L^{(n)}(\Bx)$ and $\BC_L^{(n)}(\By)$ with $|\Bx-\By|>7\,nL$. Then
\begin{equation}
\varPi\BC_L^{(n)}(\Bx)\cap\varPi\BC_L^{(n)}(\By)=\varnothing.
\end{equation}
\end{lemma}

Given an $n$-particle cube $\BC_{L_{k+1}}^{(n)}(\Bu)$ and $E\in\DR$, we denote
\begin{itemize}
\item
by $M_{\pai}^{\sep}(\BC_{L_{k+1}}^{(n)}(\Bu),E)$ the maximal number of pairwise separable, 
$(E,m_n)$-singular PI cubes $\BC_{L_k}^{(n)}(\Bu^{(j)})\subset\BC_{L_{k+1}}^{(n)}(\Bu)$;
\item
by  $M_{\pai}(\BC_{L_{k+1}}^{(n)}(\Bu),E)$ the maximal number of (not necessarily separable)
$(E,m_n)$-singular PI cubes $\BC^{(n)}_{L_k}(\Bu^{(j)})$ contain in $\BC^{(n)}_{L_{k+1}}(\Bu)$ with $|\Bu^{(j)}-\Bu^{(j')}|>7NL_k$ for all $j\neq j'$;
\item
by $M_{\fui}(\BC_{L_{k+1}}^{(n)}(\Bu),E)$ the maximal number of 
$(E,m_n)$-singular FI cubes  $\BC_{L_k}^{(n)}(\Bu^{(j)})\subset \BC_{L_{k+1}}^{(n)}(\Bu)$ with $|\Bu^{(j)}-\Bu^{(j')}|>7NL_k$ for all $j\neq j'$\footnote{Note that by lemma \ref{lem:disjointness}, two FI cubes $\BC^{(n)}_{L_k}(\Bu^{(j)})$ and $\BC^{(n)}_{L_k}(\Bu^{(j')})$ with $|\Bu^{(j)}-\Bu^{(j')}|>7NL_k $ are automatically separable.},
\item
$M_{\pai}(\BC_{L_{k+1}}^{(n)}(\Bu),I):=\sup_{E\in I}M_{\pai}(\BC_{L_{k+1}}^{(n)}(\Bu),E)$.
\item
$M_{\fui}(\BC_{L_{k+1}}^{(n)}(\Bu),I):=\sup_{E\in I}M_{\fui}(\BC_{L_{k+1}}^{(n)}(\Bu),E)$.
\item
by $M(\BC_{L_{k+1}}^{(n)}(\Bu),E)$ the maximal number of 
$(E,m_n)$-singular cubes  $\BC_{L_k}^{(n)}(\Bu^{(j)})\subset \BC_{L_{k+1}}^{(n)}(\Bu)$ with $\dist(\Bu^{(j)},\partial^-\BC^{(n)}_{L_{k+1}}(\Bu))\geq 2L_k$ and $|\Bu^{(j)}-\Bu^{(j')}|>7NL_k$ for all $j\neq j'$.
\item
by $M^{\sep}(\BC_{L_{k+1}}^{(n)}(\Bu),E)$ the maximal number of pairwise separable
$(E,m_n)$-singular cubes  $\BC_{L_k}^{(n)}(\Bu^{(j)})\subset \BC_{L_{k+1}}^{(n)}(\Bu)$ 
\end{itemize}
Clearly
\[
M_{\pai}(\BC_{L_{k+1}}^{(n)}(\Bu),E)+M_{\fui}(\BC_{L_{k+1}}^{(n)}(\Bu),E)\geq M(\BC_{L_{k+1}}^{(n)}(\Bu),E).
\]

\subsection{Pairs of partially interactive cubes}\label{ssec:PI.cubes}
Let  $\BC_{L_{k+1}}^{(n)}(\Bu)=\BC^{(n')}_{L_{k+1}}(\Bu')\times\BC^{(n'')}_{L_{k+1}}(\Bu'')$ be a PI-cube. We also write $\Bx=(\Bx',\Bx'')$ for any point $\Bx\in\BC_{L_{k+1}}^{(n)}(\Bu)$, in the same way as $\Bu=(\Bu',\Bu'')$. So  the corresponding Hamiltonian $\BH^{(n)}_{\BC_{L_{k+1}}^{(n)}(\Bu)}$ is written in the form:
\begin{equation}\label{eq:decomp,H}
\BH_{\BC_{L_{k+1}}^{(n)}(\Bu)}^{(n)}\BPsi(\Bx)=(-\BDelta\BPsi)(\Bx)+\left[\BU(\Bx')+\BV(\Bx',\omega)+\BU(\Bx'')+\BV(\Bx'',\omega)\right]\BPsi(\Bx)
\end{equation}
or, in compact form
\[
\BH_{\BC_{L_{k+1}}^{(n)}(\Bu)}^{(n)}=\BH_{\BC_{L_{k+1}}^{(n')}(\Bu')}^{(n')}\otimes\mathbf{I}+ \mathbf{I}\otimes \BH_{\BC_{L_{k+1}}^{(n'')}(\Bu'')}^{(n'')}.
\]
We denote by $\BG^{(n')}(\Bu',\Bv';E)$ and $\BG^{(n'')}(\Bu'',\Bv'';E)$ the corresponding Green functions, respectively. 

\begin{definition}\label{def:localized}
Let $n\geq 2$ and  $\BC^{(n')}_{L_k}(\Bu')\times\BC^{(n'')}_{L_k}(\Bu'')$ be the decomposition of the PI cube $\BC^{(n)}_{L_k}(\Bu)$. Then  $\BC^{(n)}_{L_k}(\Bu)$ is called
\begin{enumerate}[(i)]
\item
$m_n$-left-localized if for any $\By'\in\partial^{-}\BC^{(n')}_{L_k}(\Bu')$ and any normalized eigenfunction $\varphi^{(n')}$ of the restricted hamiltonian $\BH^{(n')}_{\BC^{(n')}_{L_k}(\Bu')}(\omega)$, we have
\[
|\varphi^{(n')}(\Bu')\varphi^{(n')}(\By')|\leq \ee^{-2\gamma(m_n,L_k,n')L_k},
\]
otherwise, it is called $m_n$-non-left-localized,
\item
$m_n$-right-localized if for any $\By''\in\partial^{-}\BC^{(n'')}_{L_k}(\Bu'')$ and any normalized eigenfunction $\phi^{(n'')}$ of the restricted hamiltonian $\BH^{(n'')}_{\BC^{(n'')}_{L_k}(\Bu'')}(\omega)$, we have
\[
|\phi^{(n'')}(\Bu'')\varphi^{(n'')}(\By'')|\leq \ee^{-2\gamma(m_n,L_k,n'')L_k},
\]
otherwise, it is called $m_n$-non-right-localized,
\item
$m_n$-localized if  it is $m_n$-left-localized and $m_n$-right-localized.
Otherwise it is called $m_n$-non-localized. 
\end{enumerate}
\end{definition}

\begin{lemma} \label{lem:localized}
Let $E\in I$ and $\BC_{L_k}^{(n)}(\Bu)$ be a PI cube. Assume that
$\BC_{L_k}^{(n)}(\Bu)$ is $E$-NR and $m_n$-localized.
Then $\BC^{(n)}_{L_k}(\Bu)$ is $(E,m_n)$-NS.
\end{lemma}

\begin{proof}
Let $\BC^{(n')}_{L_k}(\Bu')\times\BC^{(n'')}_{L_k}(\Bu'')$ be the decomposition of the PI cube $\BC^{(n)}_{L_k}(\Bu)$.
Let $\{\lambda_i,\varphi_i\}$ and $\{\mu_j,\phi_j\}$ be the eigenvalues and corresponding eigenfunctions of $\BH_{\BC_{L_k}^{(n')}(\Bu')}^{(n')}$ and $\BH_{\BC_{L_k}^{(n'')}(\Bu'')}^{(n'')}$ respectively. Then we can choose the eigenfunctions $\BPsi_{ij}$ and corresponding eigenvalues $E_{ij}$ of $\BH_{\BC^{(n)}_{L_k}(\Bu)}(\omega)$ as follows.
\[
\BPsi_{ij}=\varphi_i\otimes\phi_j,\qquad E_{ij}=\lambda_i+\mu_j.
\]
For any $\By\in\partial^{-}\BC^{(n)}_{L_k}(\Bu)$, we have that either $|\Bu'-\By'|=L_k$ or $|\Bu''-\By''|=L_k$. Hence, we can assume without lost of generality that $|\Bu''-\By''|=L_k$ if $\By\in\partial^-\BC^{(n)}_{L_k}(\Bu)$. For the Green's functions the following equation holds true since the $n$-particle eigenfunctions are tensor products of those on the underlying sub-systems of the PI cube,
\[
\BG^{(n)}_{\BC^{(n)}_{L_k}(\Bu)}(\Bu,\By;E)=\frac{\sum_{E_{ij}\varphi_i(\Bu')\varphi_i(\By')\phi_j(\Bu'')\phi_j(\By'')}}{E-E_{ij}}.
\]
Hence, since the eigenfunctions appearing in arguments and calculation are assumed normalized, we finally get
\begin{align*}
|\BG^{(n)}_{\BC^{(n)}_{L_k}(\Bu)}(\Bu,\By;E)|&\leq(2L_k+1)^{nd}\frac{\ee^{-2\gamma(m_n,L,n)L_k}}{\ee^{-L_k^{1/2}}}\\
& \leq \ee^{-\gamma(m_n,L,n)L_k},
\end{align*}
for $L_0$ and hence $L_k$ large enough.
\end{proof}
Now, before proving the main result of this subsection concerning the probability of two PI cubes to be singular at the same energy, we need first to estimate the one of a non localized cube given in the statement below. 

\begin{lemma}\label{lem:prob.localized}
Let $\BC^{(n)}_{L_k}(\Bu)$ be a PI cube. Then
\[
\prob{ \BC^{(n)}_{L_k}(\Bu) \text{ is $m_n$-non localized}}\leq \frac{1}{2}L_k^{-4p\,4^{N-n}}. 
\]
\end{lemma}
\begin{proof}
Let $\BC^{(n')}_{L_k}(\Bu')\times\BC^{(n'')}_{L_k}(\Bu'')$ be the decomposition of the PI cube $\BC^{(n)}_{L_k}(\Bu)$. Definition \ref{def:localized} implies that
\begin{align*}
\prob{\BC^{(n)}_{L_k}(\Bu) \text{ is $m_n$-non-localized}}&\leq \prob{\BC^{(n)}_{L_k}(\Bu) \text{ is $m_n$-non-left-localized}}\\
&\qquad \quad + \prob{\BC^{(n)}_{L_k}(\Bu) \text{ is $m_n$-non-right-localized}}.
\end{align*}
We estimate below only the first probability in the right hand side of the above equation since the second one can be settled similarly. Using the result of Theorem \ref{thm:EFC.bound}, the definition of the sequence $m_n$ from Definition \ref{def:m_n} and the Markov's inequality we get:
\begin{gather*}
\prob{\BC^{(n)}_{L_k}(\Bu) \text{ is $m_n$-non-left-localized}}\\
\qquad\leq \prob{\text{$ \exists \By'\in\partial^-\BC^{(n')}_{L_k}(\Bu')$, $\exists \lambda_i\in\sigma(\BH^{(n')}_{\BC^{(n')}_{L_k}(\Bu')})$: $|\varphi_i(\Bu')\varphi_i(\By')|>\ee^{-2\gamma(m_n,L_k,n')L_k}$}}\\
\qquad\leq 2\cdot(2L_k+1)^{2nd}\cdot\frac{\esm{\mathbf{\Upsilon}^{(n')}_{\BC^{(n')}_{L_k}(\Bu')}(\Bu',\By',\omega)}}{\ee^{-2\gamma(m_n,L_k,n')L_k}}\\
\qquad\leq2\cdot(2L_k+1)^{2nd}\cdot \ee^{-(\widetilde{\mu}_{n'}-2\gamma(m_n,L_k,n'))L_k}\\
\qquad\leq L_k^{-p\,4^{N-n'}}\\
\qquad\leq \frac{1}{4} L_k^{-p\,4^{N-(n-1)}}<\frac{1}{4} L_k^{-4p\,4^{N-n}},
\end{gather*}
for $L_0$ and hence $L_k$ large enough, since $2\gamma(m_n,L_k,n')<2^{N+1}m_n\leq \widetilde{\mu}_{n'}$. Finally, the sum of the two probabilities leads to the required result.
\end{proof}

Now, we state the main result of this subsection, i.e., the probability bound of two PI cubes to be singular at the same energy belonging to the  compact interval $I_0$ introduced at the beginning of the section.

\begin{theorem}\label{thm:partially.interactive}
Let $2\leq n\leq N$. There exists $L_1^*=L_1^*(N,d)>0$ such that if $L_0\geq L_1^*$ and if for $k\geq 0$ $\dsn{k,n'}$ holds true for any $1\leq n'<n$, then $\dskonn$ holds true
for any pair of separable PI cubes $\BC_{L_{k+1}}^{(n)}(\Bx)$ and $\BC_{L_{k+1}}^{(n)}(\By)$.
\end{theorem}

\begin{proof}
Let $\BC_{L_{k+1}}^{(n)}(\Bx)$ and $\BC_{L_{k+1}}^{(n)}(\By)$ be two separable PI-cubes. Consider the events:
\begin{align*}
\rB_{k+1}
&=\bigl\{\exists\,E\in I_0:\BC_{L_{k+1}}^{(n)}(\Bx)\text{ and $\BC_{L_{k+1}}^{(n)}(\By)$ are $(E,m_n)$-S}\bigr\},\\
\rR
&=\bigl\{\exists\,E\in I_0:\text{ $\BC_{L_{k+1}}^{(n)}(\Bx)$ and $\BC_{L_{k+1}}^{(n)}(\By)$ are $E$-R}\bigr\},\\
\mathcal{N}_{\Bx}
&=\bigl\{ \BC_{L_{k+1}}^{(n)}(\Bx)\text{ is $m_n$-non-localized}\bigr\},\\
\mathcal{N}_{\By}
&=\bigl\{\BC_{L_{k+1}}^{(n)}(\By)\text{ is $m_n$-non-localized}\bigr\}.
\end{align*}
If $\omega\in\rB_{k+1}\setminus\rR$, then $\forall E\in I_0$, $\BC_{L_{k+1}}^{(n)}(\Bx)$ or $\BC_{L_{k+1}}^{(n)}(\By)$ is $E$-NR. If $\BC_{L_{k+1}}^{(n)}(\By)$ is $E$-NR, then it must be $m_n$-non-localized: otherwise it would have been $(E,m_n)$-NS by Lemma~\ref{lem:localized}. Similarly, if $\BC_{L_{k+1}}^{(n)}(\Bx)$ is $E$-NR, then it must be $m_n$-non-localized. This implies that
\[
\rB_{k+1}\subset\rR\cup\mathcal{N}_{\Bx}\cup\mathcal{N}_{\By}.
\]
Therefore, using Theorem \ref{thm:Wegner} and Lemma \ref{lem:prob.localized}, we have
\begin{align*}
\DP\left\{\rB_{k+1}\right\}&\leq\DP\left\{\rR\right\}+\DP\{\mathcal{N}_{\Bx}\}+\DP\{\mathcal{N}_{\By}\}\\
&\leq L_{k+1}^{-p4^{N}}+\frac{1}{2}L_{k+1}^{-4p\,4^{N-n}}+\frac{1}{2}L_{k+1}^{-4p\,4^{N-n}}.\\
\end{align*}
 Finally 
\begin{equation}\label{eq:bound.PI}
\prob{\rB_{k+1}}\leq L_{k+1}^{-4^N\,p}+L_{k+1}^{-4p4^{N-n}}<L_{k+1}^{-2p4^{N-n}},
\end{equation}
which proves the result.
\end{proof}

For subsequent calculations and proofs we prove the following two Lemmas.

\begin{lemma}\label{lem:MPI}
If $M(\BC^{(n)}_{L_{k+1}}(\Bu),E)\geq \kappa(n)+2$ with $\kappa(n)=n^n$, then $M^{\sep}(\BC^{(n)}_{L_{k+1}}(\Bu),E)\geq 2$.\\
\noindent Similarly, if $M_{\pai}(\BC^{(n)}_{L_{k+1}}(\Bu),E)\geq \kappa(n)+2$  then $M_{\pai}^{\sep}(\BC^{(n)}_{L_{k+1}}(\Bu),E)\geq 2$.
\end{lemma}

\begin{proof}
Assume that $M^{\sep}(\BC^{(n)}_{L_{k+1}}(\Bu),E)<2$, (i.e., there is no pair of separable  cubes of radius $L_k$ in $\BC^{(n)}_{L_{k+1}}(\Bu)$), but $M(\BC^{(n)}_{L_{k+1}}(\Bu),E)\geq \kappa(n)+2$. Then $\BC^{(n)}_{L_{k+1}}(\Bu)$ must contain at least $\kappa(n)+2$ cubes $\BC^{(n)}_{L_k}(\Bv_i)$, $0\leq i\leq \kappa(n)+1$ which are non separable but satisfy $|\Bv_i-\Bv_{i'}|>7NL_k$, for all $i\neq i'$. On the other hand, by Lemma \ref{lem:separable.distant} there are at most $\kappa(n)$ cubes $\BC^{(n)}_{2nL_k}(\By_i)$, such that any cube $\BC^{(n)}_{L_k}(\Bx)$ with $\Bx\notin \bigcup_{j} \BC^{(n)}_{2nL_k}(\By_j)$ is separable from $\BC^{(n)}_{L_k}(\Bv_0)$. Hence $\Bv_i\in \bigcup_{j} \BC^{(n)}_{2nL_k}(\By_j)$ for all $i=1,\ldots,\kappa(n)+1$.  But since for all $i\neq i'$, $|\Bv_i-\Bv_{i'}|>7NL_k$, there must be at most one center $\Bv_i$ per cube $\BC^{(n)}_{2nL_k}(\By_j)$, $1\leq j\leq \kappa(n)$. Hence we come to a contradiction:
\[
\kappa(n)+1\leq \kappa(n).
\]
 The same analysis holds true if we consider only PI cubes.
\end{proof}

\begin{lemma}\label{lem:MND}
With the above notations, assume that $\dsn{k-1,n'}$ holds true for all $1\leq n'<n$ then
\begin{equation}\label{eq:MND}
\DP\left\{M_{\pai}(\BC_{L_{k+1}}^{(n)}(\Bu),I)\geq \kappa(n)+ 2\right\}\leq \frac{3^{2nd}}{2} L_{k+1}^{2nd}\left(L_k^{-4^Np}+L_k^{-4p\,4^{N-n}}\right).
\end{equation}
\end{lemma}

\begin{proof}
Suppose that $M_{\pai}(\BC^{(n)}_{L_{k+1}}(\Bu),I)\geq \kappa(n)+2$, then by Lemma \ref{lem:MPI}, $M_{\pai}^{\sep}(\BC^{(n)}_{L_{k+1}}(\Bu), I)\geq 2$, i.e., there are at least two separable $(E,m_n)$-singular PI cubes $\BC^{(n)}_{L_k}(\Bu^{(j_1)})$, $\BC^{(n)}_{L_k}(\Bu^{(j_2)})$ inside $\BC^{(n)}_{L_{k+1}}(\Bu)$.
The  number of possible pairs of centers $\{\Bu^{(j_1)},\Bu^{(j_2)}\}$ such that
\[
\BC_{L_k}^{(n)}(\Bu^{(j_1)}),\,\BC_{L_k}^{(n)}(\Bu^{(j_2)})\subset\BC_{L_{k+1}}^{(n)}(\Bu)
\]
 is bounded by $\frac{3^{2nd}}{2}L_{k+1}^{2nd}$. Then, setting 
\[
\rB_k=\{\text{$\exists E\in I$, $\BC^{(n)}_{L_k}(\Bu^{(j_1)})$, $\BC^{(n)}_{L_k}(\Bu^{(j_2)})$ are $(E,m_n)$-S}\},
\]
\[
\DP\left\{M_{\pai}^{\sep}(\BC_{L_{k+1}}^{(n)}(\Bu),I)\geq 2\right\}\leq\frac{3^{2nd}}{2}L_{k+1}^{2nd}\times\prob{\rB_k}
\]
with  $\prob{\rB_k}\leq L_k^{-4^Np}+L_k^{-4p\,4^{N-n}}$ by \eqref{eq:bound.PI}.
Here $\rB_k$ is defined as in Theorem \ref{thm:partially.interactive}.

\end{proof}

The two next subsections overlap with our earlier work \cite{E12}.

\subsection{Pairs of fully interactive cubes}\label{ssec:FI.cubes}
Our aim now is to prove $\dskonn$ for a pair of separable fully interactive cubes $\BC_{L_{k+1}}^{(n)}(\Bx)$ and $\BC_{L_{k+1}}^{(n)}(\By)$. We recall a very crucial and hard result obtained in the paper \cite{E12} and which generalized to multi-particle systems some previous work by von Dreifus and Klein \cite{DK89} in the single-particle configuration case.
\begin{lemma}\label{lem:CNR.NS}[\cite{E12} Lemma 8]
Let $J=\kappa(n)+5$ with $\kappa(n)=n^n$ and $E\in\DR$. Suppose that
\begin{enumerate}[\rm(i)]
\item
$\BC_{L_{k+1}}^{(n)}(\Bu)$ is $E$-CNR,
\item
$M(\BC_{L_{k+1}}^{(n)}(\Bu),E)\leq J$.
\end{enumerate}
Then there exists $\tilde{L}_2^*(J,N,d)>0$ such that if $L_0\geq \tilde{L}_{2}^*(J,N,d)$ we have that $\BC^{(n)}_{L_{k+1}}(\Bu)$ is $(E,m_n)$-NS.
\end{lemma}

The main result of this subsection is Theorem \ref{thm:fully.interactive}. We will  need the following preliminary results.

\begin{lemma}\label{lem:MD}
Given $k\geq0$, assume that property $\dsknn$ holds true for all pairs of separable FI cubes. Then for any $\ell\geq 1$
\begin{equation}\label{eq:MD}
\DP\left\{M_{\fui}(\BC_{L_{k+1}}^{(n)}(\Bu),I_0)\geq 2\ell\right\}\leq C(n,N,d,\ell)L_k^{2\ell dn\alpha}L_k^{-2\ell p\,4^{N-n}}.
\end{equation}
\end{lemma}

\begin{proof}
Suppose there exist $2\ell$ pairwise separable, fully interactive cubes $\BC_{L_k}^{(n)}(\Bu^{(j)})$ $\subset\BC_{L_{k+1}}^{(n)}(\Bu)$, $1\leq j\leq 2\ell$. Then, by Lemma \ref{lem:disjointness}, for any pair $\BC_{L_k}^{(n)}(\Bu^{(2i-1)})$, $\BC_{L_k}^{(n)}(\Bu^{(2i)})$, the corresponding random Hamiltonians $\BH_{\BC_{L_k}^{(n)}(\Bu^{(2i-1)})}^{(n)}$ and $\BH^{(n)}_{\BC_{L_k}^{(n)}(\Bu^{(2i)})}$ are independent, and so are their spectra and their Green functions. For $i=1,\dots,\ell$ we consider the events:
\[
\rA_i=\left\{\exists\,E\in I_0:\BC_{L_k}^{(n)}(\Bu^{(2i-1)})\text{ and $\BC_{L_k}^{(n)}(\Bu^{(2i)})$ are $(E,m_n)$-S}\right\}.
\]
Then by assumption $\dsknn$, we have, for $i=1,\dots,\ell$,
\begin{equation}
\DP\left\{\rA_i\right\}\leq L_k^{-2p\,4^{N-n}},
\end{equation}
and, by independence of events $\rA_1,\dots,\rA_{\ell}$,
\begin{equation}
\DP\Bigl\{\bigcap_{1\leq i\leq\ell}\rA_i\Bigr\}=\prod_{i=1}^{\ell}\DP(\rA_i)\leq\bigl(L_k^{-2p\,4^{N-n}}\bigr)^{\ell}.
\end{equation}
To complete the proof, note that the total number of different families of $2\ell$ cubes $\BC_{L_k}^{(n)}(\Bu^{(j)})\subset\BC_{L_{k+1}}^{(n)}(\Bu)$, $1\leq j\leq 2\ell$, is bounded by
\[
\frac{1}{(2\ell)!}\left|\BC_{L_{k+1}}^{(n)}(\Bu)\right|^{2\ell}\leq C(n,N,\ell,d)L_{k}^{2\ell dn\alpha}.\qedhere
\]
\end{proof}

\begin{theorem}\label{thm:fully.interactive}
Let $1\leq n\leq N$. There exists $L_2^*=L_2^*(N,d)>0$ such that if $L_0\geq L_2^*$ and if for $k\geq 0$
\begin{enumerate}[\rm(i)]
\item
$\dsn{k-1,n'}$ for all $1\leq n'<n$ holds true,
\item
$\dsn{k,n}$ holds true for all pairs of FI cubes,
\end{enumerate}
then $\dskonn$ holds true
for any pair of separable FI cubes $\BC_{L_{k+1}}^{(n)}(\Bx)$ and $\BC_{L_{k+1}}^{(n)}(\By)$.
\end{theorem}

Above, we used the convention that $\dsn{-1,n}$ means no assumption.

\begin{proof}
Consider a pair of separable FI cubes $\BC_{L_{k+1}}^{(n)}(\Bx)$, $\BC_{L_{k+1}}^{(n)}(\By)$ and set $J=\kappa(n)+5$. Define
\begin{align*}
\rB_{k+1}
&=\left\{\exists\,E\in I_0:\BC_{L_{k+1}}^{(n)}(\Bx)\text{ and $\BC_{L_{k+1}}^{(n)}(\By)$ are $(E,m_n)$-S}\right\},\\
\Sigma
&=\left\{\exists\,E\in I_0:\text{neither $\BC_{L_{k+1}}^{(n)}(\Bx)$ nor $\BC_{L_{k+1}}^{(n)}(\By)$ is $E$-CNR}\right\},\\
\rS_{\Bx}
&=\left\{\exists\,E\in I_0:M(\BC_{L_{k+1}}^{(n)}(\Bx);E)\geq J+1\right\},\\
\rS_{\By}
&=\left\{\exists\,E\in I_0:M(\BC_{L_{k+1}}^{(n)}(\By),E)\geq J+1\right\}.
\end{align*}
Let $\omega\in\rB_{k+1}$. If $\omega\notin\Sigma\cup\rS_{\Bx}$, then $\forall E\in I_0$ either $\BC_{L_{k+1}}^{(n)}(\Bx)$ or $\BC_{L_{k+1}}^{(n)}(\By)$ is $E$-CNR and $M(\BC_{L_{k+1}}^{(n)}(\Bx),E)\leq J$. The cube $\BC_{L_{k+1}}^{(n)}(\Bx)$ cannot be $E$-CNR: indeed, by Lemma \ref{lem:CNR.NS} it would be $(E,m_n)$-NS. So the cube $\BC_{L_{k+1}}^{(n)}(\By)$ is $E$-CNR and $(E,m_n)$-S. This implies again by Lemma \ref{lem:CNR.NS} that
\[
M(\BC_{L_{k+1}}^{(n)}(\By),E)\geq J+1.
\]
Therefore $\omega\in\rS_{\By}$, so that $\rB_{k+1}\subset\Sigma\cup\rS_{\Bx}\cup\rS_{\By}$, hence
\[
\DP\left\{\rB_{k+1}\right\}
\leq\DP\{\Sigma\}+\DP\{\rS_{\Bx}\}+\DP\{\rS_{\By}\}
\]
and $\prob{\Sigma}\leq L_{k+1}^{-4^N\,p}$ by Theorem \ref{thm:Wegner}.
Now let us estimate $\DP\{\rS_{\Bx}\}$ and similarly $\DP\{\rS_{\By}\}$. Since 
\[
M_{\pai}(\BC_{L_{k+1}}^{(n)}(\Bx),E)+M_{\fui}(\BC_{L_{k+1}}^{(n)}(\Bx),E)\geq M(\BC_{L_{k+1}}^{(n)}(\Bx),E),
\] 
the inequality $M(\BC_{L_{k+1}}^{(n)}(\Bx),E)\geq\kappa(n)+ 6$, implies that either $M_{\pai}(\BC_{L_{k+1}}^{(n)}(\Bx),E)\geq\kappa(n)+ 2$, or $M_{\fui}(\BC_{L_{k+1}}^{(n)}(\Bx),E)\geq 4$. Therefore, by Lemma \ref{lem:MND} and Lemma \ref{lem:MD} (with $\ell=2$),
\begin{align*}
\DP\{\rS_{\Bx}\}&\leq\DP\left\{\exists\,E\in I:M_{\pai}(\BC_{L_{k+1}}^{(n)}(\Bx),E)\geq \kappa(n)+2\right\}\\
&\quad+\DP\left\{\exists\,E\in I:M_{\fui}(\BC_{L_{k+1}}^{(n)}(\Bx),E)\geq 4\right\}\\
&\leq\frac{3^{2nd}}{2}L_{k+1}^{2nd}(L_k^{-4^Np}+L_k^{-4p\,4^{N-n}})+C'(n,N,d)L_{k+1}^{4 dn-\frac{4 p}{\alpha}4^{N-n}}\\
&\leq C''(n,N,d)\left(L_{k+1}^{-\frac{4^Np}{\alpha}+2nd}+L_{k+1}^{-\frac{4p}{\alpha}4^{N-n}+2nd}+L_{k+1}^{-\frac{4p}{\alpha}4^{N-n}+4nd}\right)\\
&\leq C'''(n,N,d) L_{k+1}^{-\frac{4p}{\alpha}4^{N-n}+4nd}\tag{$\alpha=3/2$}\\
&\leq\frac{1}{4}L_{k+1}^{-2p\,4^{N-n}},
\end{align*}
where we used $p>4\alpha Nd=6Nd$. Finally 
\[ 
\prob{\rB_{k+1}}\leq L_{k+1}^{-4^Np}+\frac{1}{2}L_{k+1}^{-2p4^{N-n}}<L_{k+1}^{-2p4^{N-n}}.
\]
\end{proof}

\subsection{Mixed pairs of cubes}\label{ssec:mixed.S}
Finally, it remains only to derive $\dskonn$ in case (III), i.e., for pairs of $n$-particle cubes where one is PI while the other is FI.

\begin{theorem}\label{thm:mixed}
Let $1\leq n\leq N$. There exists $L_3^*=L_3^*(N,d)>0$ such that if $L_0\geq L_3^*(N,d)$ and if for $k\geq 0$,
\begin{enumerate}[\rm(i)]
\item
$\dsn{k-1,n'}$ holds true for all $1\leq n'<n$,
\item
$\dsn{k,n'}$ holds true for all $1\leq n'<n$ and
\item
$\dsknn$ holds true for all pairs of FI cubes,
\end{enumerate}
then $\dskonn$ holds true for any pair of separable cubes $\BC_{L_{k+1}}^{(n)}(\Bx)$, $\BC_{L_{k+1}}^{(n)}(\By)$ where one is PI while the other is FI.
\end{theorem}

\begin{proof}
Consider a pair of separable $n$-particle cubes $\BC_{L_{k+1}}^{(n)}(\Bx)$, $\BC_{L_{k+1}}^{(n)}(\By)$ and suppose that $\BC_{L_{k+1}}^{(n)}(\Bx)$ is PI while $\BC_{L_{k+1}}^{(n)}(\By)$ is FI. Set $J=\kappa(n)+5$ and introduce the events
\begin{align*}
\rB_{k+1}
&=\left\{\exists\,E\in I_0:\BC_{L_{k+1}}^{(n)}(\Bx)\text{ and $\BC_{L_{k+1}}^{(n)}(\By)$ are $(E,m)$-S}\right\},\\
\Sigma
&=\left\{\exists\,E\in I_0: \BC_{L_{k+1}}^{(n)}(\Bx)\ \text{is not $E$-CNR and}\ \BC_{L_{k+1}}^{(n)}(\By)\text{ is not $E$-CNR}\right\},\\
\mathcal{N}_{\Bx}
&=\left\{ \BC_{L_{k+1}}^{(n)}(\Bx)\text{ is $m_n$-non-localized}\right\},\\
\rS_{\By}
&=\left\{ \exists\,E\in I_0:M(\BC_{L_{k+1}}^{(n)}(\By),E)\geq J+1\right\}.
\end{align*}
Let $\omega\in \rB_{k+1}\setminus (\Sigma\cup \mathcal{N}_{\Bx})$, then for all $E\in I_0$ either $\BC_{L_{k+1}}^{(n)}(\Bx)$ is $E$-CNR or $\BC_{L_{k+1}}^{(n)}(\By)$ is $E$-CNR and  $\BC_{L_{k+1}}^{(n)}(\Bx)$ is $m_n$-localized. The cube $\BC_{L_{k+1}}^{(n)}(\Bx)$ cannot be $E$-CNR. Indeed, by Lemma \ref{lem:localized} it would have been $(E,m_n)$-NS. Thus the cube $\BC_{L_{k+1}}^{(n)}(\By)$ is $E$-CNR, so by Lemma \ref{lem:CNR.NS}, $M(\BC_{L_{k+1}}^{(n)}(\By);E)\geq J+1$: otherwise $\BC_{L_{k+1}}^{(n)}(\By)$ would be $(E,m_n)$-NS. Therefore $\omega\in \rS_{\By}$. Consequently,
\[
\rB_{k+1}\subset \Sigma \cup \mathcal{N}_{\Bx}\cup\rS_{\By}.
\]
Recall that the probabilities $\DP\{\mathcal{N}_{\Bx}\}$ and $\DP\{\rS_{\By}\}$ have already been estimated in Sections \ref{ssec:PI.cubes} and \ref{ssec:FI.cubes}. We therefore obtain
\begin{align*}
\DP\left\{\rB_{k+1}\right\}&\leq \DP\{\Sigma\}+\DP\{\mathcal{N}_{\Bx}\}+\DP\{\rS_{\By}\}\\
&\leq L_{k+1}^{-4^Np}+ \frac{1}{2}L_{k+1}^{-4p\,4^{N-n}}+ \frac{1}{4}L_{k+1}^{-2p\,4^{N-n}}\leq L_{k+1}^{-2p\,4^{N-n}}.\qedhere
\end{align*}
\end{proof}

\section{Conclusion: the multi-particle multi-scale analysis}

\begin{theorem}\label{thm:DS.k.N}
Let $1\leq n\leq N$ and  $\BH^{(n)}(\omega)=-\BDelta+\sum_{j=1}^n V(x_j,\omega)+\BU$, where $\BU$, $V$ satisfy
 $\condI$ and $\condP$  respectively. There exists $m_n>0$ such that for any $p>6Nd$ property $\dsknn$ holds true for all $k\geq 0$ provided $L_0$ is large enough.
\end{theorem}
\begin{proof}
We prove that for each $n=1,\ldots,N$, property $\dsknn$ is valid. To do so, we use an induction on the number of particles $n'=1,\ldots,n$. For $n=1$  property $\dsn{k,1}$ holds true for all $k\geq 0$ by  the single-particle localization theory (\cites{DK89,K08}). Now suppose that for all $1\leq n'<n$, $\dsn{k,n'}$ holds true for all	 $k\geq 0$, we aim to prove that $\dsknn$ holds true for all $k\geq 0$. For $k=0$, $\dsn{0,n}$ is valid using Theorem \ref{thm:initial.var.energy}. Next, suppose that $\dsn{k',n}$ holds true for all $k'<k$, then by  combining this last assumption with $\dsn{k,n'}$ above, one can conclude that
\begin{enumerate}
\item[\rm(i)] $\dsknn$ holds true for all $k\geq 0$ and for all pairs of PI cubes using Theorem \ref{thm:partially.interactive},
\item[\rm(ii)] $\dsknn$ holds true for all $k\geq 0$ and for all pairs of FI cubes using Theorem \ref{thm:fully.interactive},
\item[\rm(iii)] $\dsknn$ holds true for all $k\geq 0$ and for all pairs of MI cubes using Theorem \ref{thm:mixed}.
\end{enumerate}
 Hence Theorem \ref{thm:DS.k.N} is proven.
\end{proof}

\section{Proofs of the localization results}
For any number $N\geq2$ of particles on the lattice, we use property $\dsn{k,N}$ established  in the above sections to prove our main localization results: Theorems \ref{thm:exp.loc} and \ref{thm:dynamical.loc} (the spectral exponential and the strong Hilbert-Schmidt dynamical localization), in the same way as in our previous paper \cite{E12}, section 5 and in each interval $I_0=[E_0-\delta;E_0+\delta]$ with $E_0\in I$. Finally, since the interval $I$ is compact, we also get localization on $I$.

\end{document}